\documentclass{article}
\usepackage[utf8]{inputenc}
\usepackage{amsthm}
\usepackage{amssymb}
\usepackage{amsmath}
\usepackage{amsfonts}
\usepackage{fullpage}

\usepackage[backref = page]{hyperref}
\hypersetup{
    colorlinks = true,
    citecolor = blue,
    linkcolor = blue
}
\usepackage[capitalise]{cleveref}

\usepackage{mathtools}

\usepackage[textsize=tiny]{todonotes}

\usepackage{booktabs}
\usepackage{multirow}

\newcommand{\BPCC}{\mathrm{BPP}^{\text{cc}}}
\newcommand{\NPCC}{\mathrm{NP}^{\text{cc}}}

\newcommand{\AND}{\text{AND}}
\newcommand{\OR}{\text{OR}}
\newcommand{\x}{\mathbf{x}}
\newcommand{\y}{\mathbf{y}}

\newcommand{\D}{\mathcal{D}}
\newcommand{\cF}{\mathcal{F}}

\newcommand{\cW}{\mathcal{W}}

\newcommand{\F}{\mathbb{F}}
\newcommand{\R}{\mathbb{R}}

\usepackage{thmtools} 
\usepackage{thm-restate}
\usepackage{cite}
\usepackage{enumerate}

\Crefname{claim}{Claim}{Claims}

\declaretheoremstyle[
    spaceabove=-6pt, 
    spacebelow=12pt, 
    headfont=\normalfont\bfseries, 
    bodyfont = \normalfont,
    postheadspace=1em
    ]{examplestyle} 

\declaretheorem[name=Theorem,numberwithin=section]{theorem}
\declaretheorem[name=Lemma,sibling=theorem]{lemma}
\declaretheorem[name=Claim,sibling=theorem]{claim}
\declaretheorem[name=Remark,sibling=theorem]{remark}
\declaretheorem[name=Corollary,sibling=theorem]{corollary}
\declaretheorem[name=Definition,sibling=theorem]{definition}
\declaretheorem[name=Example,sibling=theorem]{example}
\declaretheorem[name=Conjecture,sibling=theorem]{conjecture}
\newcommand{\andFunction}[1]{{{#1}_{\land}}}
\newcommand{\xorFunction}[1]{{{#1}_{\lxor}}}
\newcommand{\lxor}{\oplus}
\newcommand{\biglor}{\bigvee}
\newcommand{\bigland}{\bigwedge}

\newcommand{\UDISJ}[1]{\mathrm{UDISJ}_{#1}}
\newcommand{\EQ}[1]{\mathrm{EQ}_{#1}}
\newcommand{\IP}[1]{\mathrm{IP}_{#1}}
\newcommand{\Ind}[1]{\mathrm{Ind}_{#1}}

\newcommand{\bool}{\{0, 1\}}
\newcommand{\set}[1]{\left\{{#1}\right\}}

\newcommand{\PDT}[1][ ]{\mathrm{P}^{\text{dt}}\ifthenelse{\equal{#1}{ }}{ }{(#1)}}
\newcommand{\BPPDT}[1][ ]{\mathrm{BPP}^{\text{dt}}\ifthenelse{\equal{#1}{ }}{ }{(#1)}}
\newcommand{\PandDT}[1][ ]{\mathrm{P}^{\land\text{-dt}}\ifthenelse{\equal{#1}{ }}{ }{(#1)}}
\newcommand{\PxorDT}[1][ ]{\mathrm{P}^{\lxor\text{-dt}}\ifthenelse{\equal{#1}{ }}{ }{(#1)}}
\newcommand{\PzeroDT}[1][ ]{\mathrm{P}^{0\text{-dt}}\ifthenelse{\equal{#1}{ }}{ }{(#1)}}

\newcommand{\HSC}[1][ ]{\mathrm{HSC}\ifthenelse{\equal{#1}{ }}{ }{(#1)}}
\newcommand{\MBS}[1][ ]{\mathrm{MBS}\ifthenelse{\equal{#1}{ }}{ }{(#1)}}
\newcommand{\FHSC}[1][ ]{\mathrm{FHSC}\ifthenelse{\equal{#1}{ }}{ }{(#1)}}
\newcommand{\FMBS}[1][ ]{\mathrm{FMBS}\ifthenelse{\equal{#1}{ }}{ }{(#1)}}

\newcommand{\spar}[1][ ]{\mathrm{spar}\ifthenelse{\equal{#1}{ }}{ }{(#1)}}

\newcommand{\rank}[1][ ]{\mathrm{rank}\ifthenelse{\equal{#1}{ }}{ }{(#1)}}
\newcommand{\PCC}[1][ ]{\mathrm{P}^{\text{cc}}\ifthenelse{\equal{#1}{ }}{ }{(#1)}}
\newcommand{\BPPCC}[1][ ]{\mathrm{BPP}^{\text{cc}}\ifthenelse{\equal{#1}{ }}{ }{(#1)}}

\newcommand{\poly}{\mathrm{poly}}
\newcommand{\ceil}[1]{\left\lceil {#1} \right\rceil}
\newcommand{\mon}[1][ ]{\mathcal{M}(#1)}

\title{Log-rank and lifting for AND-functions}

\author{
Alexander Knop
\thanks{Mathematics, University of California, San Diego. Email: \textit{aaknop@gmail.com}}
\and
Shachar Lovett
\thanks{Computer Science and Engineering, University of California, San Diego. Research supported by NSF award 2006443. Email: \textit{slovett@ucsd.edu}}
\and 
Sam McGuire
\thanks{Computer Science and Engineering, University of California, San Diego. Research supported by NSF award 1909634 and a Simons Investigator award. Email: \textit{shmcguir@eng.uscd.edu}}
\and
Weiqiang Yuan
\thanks{Institute for Interdisciplinary Information Sciences, Tsinghua University. Email: \textit{yuanwq17@mails.tsinghua.edu.cn}}}

\begin{document}

\maketitle

\begin{abstract}
    Let $f: \{0, 1\}^n \to \{0, 1\}$ be a boolean function, and let 
$f_\land(x, y) = f(x \land y)$ denote the \emph{AND-function} of 
$f$, where $x \land y$ denotes bit-wise AND. We study the 
deterministic communication complexity of $f_{\land}$
and show that, up to a $\log n$ factor, it is bounded by a polynomial
in the logarithm of the real rank of the communication matrix of $f_{\land}$.
This comes within a $\log n$ factor of establishing the log-rank conjecture
for AND-functions with \emph{no assumptions} on $f$. Our result stands in
contrast with previous results on special cases of the log-rank conjecture, 
which needed significant restrictions on $f$ such as monotonicity or low 
$\mathbb{F}_2$-degree. Our techniques can also be used to 
prove (within a $\log n$ factor) a \emph{lifting theorem} for AND-functions, 
stating that the deterministic communication complexity of $f_{\land}$ is
polynomially related to the \emph{AND-decision tree complexity} of $f$. 

The results rely on a new structural result regarding boolean functions 
$f: \{0, 1\}^n \to \{0, 1\}$ with a sparse polynomial representation, which
may be of independent interest. We show that if the polynomial computing $f$
has few monomials then the set system of the monomials has a small hitting 
set, of size poly-logarithmic in its sparsity. We also establish extensions 
of this result to multi-linear polynomials $f: \{0, 1\}^n \to \R$ with a 
larger range. 
\end{abstract}

\section{Introduction}

Communication complexity has seen rapid development in the last 
couple of decades. However, most of the celebrated results in the field are 
about the communication complexity of important \emph{concrete functions},
such as set disjointness~\cite{razborov1992distributional} and gap Hamming 
distance~\cite{chakrabarti2012hamming}. Unfortunately, the understanding of 
communication complexity of \emph{arbitrary functions} is still lacking. 

Probably the most famous problem of this type is the log-rank 
conjecture~\cite{lovasz1988logrank}.
It speculates that given any total boolean communication problem, its
deterministic communication complexity is polynomially related to the 
logarithm of the real rank of its associated communication matrix. Currently,
there is an exponential gap between the lower and upper bounds relating to 
the log-rank conjecture. The best known upper bound~\cite{lovett2016rootrank}
states that the  communication complexity of a boolean function $F$ is at most
$O(\sqrt{\rank[F]} \log \rank[F])$, where  $\rank[F]$ denotes the real rank 
of the communication matrix of $F$. On the other hand, the best known lower 
bound~\cite{goos18detpart} states that there exist a boolean function $F$ 
with communication complexity $\Omega(\log^2(\rank[F]))$.

Given this exponential gap and lack of progress for general communication 
problems, many works~\cite{goos20019PNP,goos2020autoamisationCP,%
goos2017BPPlifting,zhang2009communication,mukhopadhyay2019lifting,%
pitassi2020hierarchies,chattopadhyay2019approxlogrank,goos18detpart,%
rezende2020liftingsimplegadgets,cChattopadhyay2019BPPlifting,%
mande2020paritylogrank,hatami2018xorfunctions,tsang2013fourier,%
montanaro2009communication,leung2011tight,zhang2014efficient,%
yao2015parity,hatami2017unbounded,sanyal2015fourieranddimension}
focused on the communication complexity
of functions with some restricted structure. In particular, the study of composed
functions was especially successful, and produced the so-called lifting method, 
which connects query complexity measures of boolean functions with communication
complexity measures of their corresponding communication problems.

Concretely, given a boolean function $f:\bool^n \to \bool$ and a \emph{gadget} 
$g : \bool^\ell \times \bool^m \to \bool$, the corresponding lifted function is 
the following communication problem: Alice gets as input $x \in (\bool^\ell)^n$, 
Bob gets as input $y \in (\bool^m)^n$, and their goal is to compute the composed
function $f \circ g^n$, defined as 
$(f \circ g^n)(x, y) = f(g(x_1, y_1), \dots, g(x_n, y_n))$. Lifting theorems allow
to connect query complexity measures of the underlying boolean function $f$ with 
communication complexity measures of the composed function. \Cref{figure:lifting} 
lists some notable examples.

\begin{figure}[ht]
    \centering
    \begin{tabular}{l l l l l}
        \toprule
        Gadget &  Query Model & Communication Model & Total Functions & Reference\\
        \midrule
        \multirow{2}{*}{$\Ind{m}$} & $\PDT$ & $\PCC$ & No & \cite{raz1999separationNC} \\
        & $\BPPDT$ & $\BPPCC$ & No &  \cite{goos2017BPPlifting} \\
        \multirow{2}{*}{$\IP{\log m}$} & $\PDT$ & $\PCC$ & No  & \cite{chattopadhyay2019lowdiscrep} \\
        & $\BPPDT$ & $\BPPCC$ & No &  \cite{cChattopadhyay2019BPPlifting} \\
        $\EQ{\log m}$ & $\PandDT$ & $\PCC$ & No  & \cite{mukhopadhyay2019lifting}
        \vspace{10pt} \\
        $\lxor$ & $\PxorDT$ & $\PCC$ & Yes & \cite{hatami2018xorfunctions} \\
        $g$ & $\deg$ & $\rank$ & Yes  &  \cite{sherstov2010quantumclassical} \\
      \bottomrule
    \end{tabular}
    \caption{%
        Query-to-communication lifting theorems. 
        The parameter $m$ is polynomial in $n$;
        $g$ in the last line is any function that has as sub-functions both an AND and an OR.
        $\PCC$ denotes determenistic communication complexity, $\PDT$ denotes decision tree 
        complexity, $\BPPDT$ denotes the probabilistic decision tree complexity with bounded 
        error, $\BPPCC$ denotes the probabilistic communication complexity with bounded error,
        $\PandDT$ denotes AND-decision tree complexity, $\deg$ denotes the real degree,
        and $\rank$ denotes the real rank.
    }
    \label{figure:lifting}
\end{figure}

Of particular interest to us are lifting theorems with very simple gadgets. The reason 
for that is twofold. First, using complex gadgets (such as inner product or indexing)
yields sub-optimal bounds in applications. A second and perhaps more important reason is 
that the study of composed functions with complex gadgets does not bring us any closer
towards the understanding of general communication problems. This is because the 
corresponding lifting theorems connect the communication complexity of the lifted function
to well-studied query measures of the underlying boolean function (such as decision tree 
complexity, or degree as a real polynomial), and hence does not shed new light on general
communication problems.

Thus, in this paper we consider gadgets which are as simple as they could be --- one-bit 
gadgets. In fact, there are only two non-equivalent one-bit gadgets: one-bit XOR, which
yields XOR-functions; and one-bit AND, which yields AND-functions. As we shortly discuss, 
they naturally correspond to query models which extend the standard ones: parity-decision 
trees and AND-decision trees.

\paragraph{XOR-functions.}
XOR-functions have been studied in several works~\cite{%
mande2020paritylogrank,%
hatami2018xorfunctions,tsang2013fourier,%
zhang2009communication,montanaro2009communication,%
leung2011tight,zhang2014efficient,yao2015parity,hatami2017unbounded,%
sanyal2015fourieranddimension}. 
Given a boolean function $f:\bool^n \to \bool$,
its corresponding XOR-function is $\xorFunction{f} = f \circ \lxor^n$, defined as
$\xorFunction{f}(x,y)=f(x \lxor y)$. A natural query measure corresponding to the 
communication complexity of XOR-functions is the \emph{Parity-Decision Tree} (PDT) model.
This model is an extension of the standard decision tree model, where nodes can query an
arbitrary parity of the bits. To see the connection, note that if $f$ can be computed by
a PDT of depth $d$ (denoted by $\PxorDT[f] = d$), then $\xorFunction{f}$ has a communication
protocol of complexity $2d$. This is by simulating the computation in the PDT: whenever
the PDT needs to compute the parity of $x \lxor y$ on some set $S$ of coordinates, each
player computes the corresponding parity on their input, and then they exchange the 
answers, which allows to compute the corresponding parity on $x \lxor y$ as well, and 
continue to traverse the tree. Thus we have $\PCC[\xorFunction{f}] \le 2 \PxorDT[f]$.

In the other direction, \cite{hatami2018xorfunctions} proved that $\PxorDT[f]$ is at most
a polynomial in the communication complexity of  $\xorFunction{f}$. That is, 
$\PxorDT[f] \le \text{poly}\left(\PCC[\xorFunction{f}] \right)$. Thus, the two measures
are equivalent, up to polynomial factors.

If one considers the log-rank conjecture for XOR-functions, then a simple 
observation~\cite{tsang2013fourier} is that the rank of the communication matrix of 
$\xorFunction{f}$ is equal to Fourier sparsity of $f$. Thus, in order to prove the 
log-rank conjecture for XOR-functions it is sufficient to show  that $\PxorDT[f]$ is at
most a polynomial in the log of the Fourier sparsity of $f$. Unfortunately, the latter 
relation is currently unknown.

\paragraph{AND-functions.}
The goal of this paper is to develop an analogous theory of AND-functions. Let 
$f : \bool^n \to \bool$  be a boolean function. Its corresponding AND-function
is $\andFunction{f} = f \circ \land^n$, defined as
$\andFunction{f}(x, y) = f(x \land y)$. Similar to the case of XOR-functions, there is
a corresponding natural query model, \emph{AND-Decision Tree} (ADT), where each node
in the decision tree can query an arbitrary AND of the input bits. We denote by 
$\PandDT[f]$ the minimal depth of an ADT computing $f$. Also here, efficient ADTs for
$f$ imply efficient communication protocols for $\andFunction{f}$, where 
$\PCC[\andFunction{f}] \le 2 \PandDT[f]$. Our main focus in this work is 
\begin{enumerate}[(i)]
    \item lifting theorems for AND-functions, and
    \item the log-rank conjecture for AND-functions.
\end{enumerate}
Concretely, we will show that assuming that $\andFunction{f}$ has either (i) 
efficient deterministic communication protocol or (ii) low rank, then $f$ has an
efficient ADT. As we will shortly see, understanding both questions is directly related 
to understanding the monomial structure of polynomials computing boolean functions.

\subsection{Main results}

Let $f : \bool^n \to \bool$ be a boolean function. It is computed by a unique multi-linear 
polynomial over the reals. That is, $f(x) = \sum_s f_s \prod_{i \in s} x_i$, where 
$s \subseteq [n]$ and $f_s \in \R$ are real-valued coefficients. The sparsity of $f$, 
denoted $\spar[f]$, is the number of nonzero coefficients in the decomposition.
This is related 
to AND-functions, as a simple observation (\Cref{claim:rank_sparsity})
is that this also equals the rank of its communication matrix, namely 
$\rank[\andFunction{f}] = \spar[f]$.

Before describing our results, we need one more definition. Let $\mathcal{F}$ be
a set system (family of sets). A set $H$ is a \emph{hitting set} for $\mathcal{F}$ if
it intersects all the sets in $\mathcal{F}$. Of particular interest to us are set systems
that correspond to the monomials of boolean functions. Given a boolean function $f$, 
define $\mon[f] = \set{s \ :\  f_s \neq 0, s \neq 0}$
to be the set system of the non-constant monomials of $f$. We exclude the constant term as
it is irrelevant for the purpose of constructing hitting sets, and it simplifies some of
the later arguments. Note that $|\mon[f]| \in \set{\spar[f], \spar[f] - 1}$.

Our main combinatorial result is that set systems corresponding to the monomials of
boolean functions have small hitting sets.

\begin{restatable}{theorem}{mainthm}
    \label{theorem:sparsity-to-hitting}
    Let $f : \bool^n \to \bool$ be a boolean function with sparsity $\spar(f) = r$. 
    Then there exists a hitting set $H$ for $\mon[f]$ of size $|H| = O((\log r)^5)$.
\end{restatable}

This result can be seen as an analog of a similar result for union-closed families. 
A set system $\mathcal{F}$ is union-closed if it is closed under taking unions; namely, 
if $S_1, S_2 \in \mathcal{F}$ then also  $S_1 \cup S_2 \in \mathcal{F}$. A famous 
conjecture of Frankl~\cite{frankl1983tracefinitesets} is that in any union-closed 
family $\mathcal{F}$ there is an element which belongs to at least half the sets in
the set system. Assume $|\mathcal{F}| = r$; the best known result in this direction is 
that $\mathcal{F}$ has a hitting set of size $\log(r)$~\cite{knill1994graph}, which 
implies that one of its elements belongs to a $1 / \log(r)$ fraction of sets in the
set system. We view \cref{theorem:sparsity-to-hitting} as a qualitative extension 
of this result to more general set systems.

Our main application of \Cref{theorem:sparsity-to-hitting} is a near-resolution of the
log-rank conjecture for AND-functions. Our bounds nearly match the conjectured bounds 
(poly-log in the rank), except for an extra $\log(n)$ factor that we are currently 
unable to eliminate.

\begin{restatable}[Log-rank Theorem for AND-functions]{theorem}{logrank}
\label{thm:logrank}
    Let $f : \bool^n \to \bool$ be a boolean function. 
    Let $r = \spar[f] = \rank[\andFunction{f}]$. Then $f$ can be computed by an 
    AND-decision tree of depth 
    \[
        \PandDT[f] = O((\log r)^5 \cdot \log n).
    \]
    In particular, the deterministic communication complexity of $\andFunction{f}$ is
    bounded by
    \[
         \PCC[\andFunction{f}] = O((\log r)^5 \cdot \log n).
    \]
\end{restatable}

Note that if $f : \bool^n \to \bool$ is a function of sparsity at least 
$n^{0.1}$, say, then \Cref{thm:logrank} proves the log-rank conjecture for its 
corresponding AND-function. Thus, the only remaining obstacle is to extend the result
to very sparse functions. 

Observe that \Cref{thm:logrank} implies a lifting theorem for AND-functions. Assume 
that $\andFunction{f}$ has deterministic communication complexity $C$. The rank of 
the associated communication matrix is then at most $2^C$, which by  \Cref{thm:logrank}
gives an ADT for $f$ of depth $O(C^5 \log n)$. We can improve the exponent $5$ to $3$
by directly exploiting the existence of a communication protocol.

\begin{restatable}[Lifting Theorem for AND-functions]{theorem}{lifting}
\label{thm:lifting}
    Let $f: \bool^n \to \bool$ be a boolean 
    function. Let $C = \PCC[\andFunction{f}]$ 
    denote the deterministic communication complexity of its corresponding AND-function.
    Then $f$ can be computed by an AND-decision tree of depth
    \[
        \PandDT[f] = O(C^3 \cdot \log n).
    \]
\end{restatable}

\subsection{Proof overview}

We first discuss how our combinatorial theorem (\Cref{theorem:sparsity-to-hitting}) 
implies the log-rank theorem (\Cref{thm:logrank}). It relies on showing that sparse 
boolean functions have efficient AND-decision trees (ADTs). 

Let $f$ be a boolean function with $\spar[f] = r$. Our goal is to construct an ADT for $f$
of depth $\poly(\log r) \cdot \log(n)$. This directly implies \Cref{thm:logrank}, as the
sparsity of $f$ equals the rank of its AND-function $\andFunction{f}$, and an ADT for $f$
of depth $d$ implies a protocol for $\andFunction{f}$ which sends $2d$ bits.

It will be convenient to first consider another model of decision trees, called 
\emph{zero decision trees}. A (standard) decision tree computing $f$ has zero 
decision tree complexity $d$, if any path from root to leaf in it queries at most $d$
variables which evaluate to $0$. We denote by $\PzeroDT[f]$ the minimal such $d$ over all
decision trees that compute $f$. It is shown in~\cite{mukhopadhyay2019lifting} 
(see also \Cref{claim:0DT_ADT}) that ADT complexity and zero DT complexity are tightly
connected. Concretely, for any boolean function $f$ they show that
\[
    \PzeroDT[f] \le \PandDT[f] \le \PzeroDT[f] \cdot \lceil \log(n+1) \rceil.
\]
Thus, we will show that $\PzeroDT[f] \le \poly(\log r)$, which implies our target bound of 
$\PandDT[f]$.

\Cref{theorem:sparsity-to-hitting} gives that there is a hitting set  size $h = \poly(\log r)$
which intersects all the monomials of $f$. In particular, there is a variable $x_i$ that 
intersects at least a $1 / h$ fraction of the monomials of $f$. The decision tree will first 
query $x_i$, and then branch depending on whether $x_i = 0$ or $x_i = 1$. We use the simple 
fact that the sparsity of $f$ cannot increase when variables are fixed, and continue this
process, until the value of the function is determined. Observe that every time that we 
query a variable and get $0$, we eliminates a $1 / h$ fraction of the monomials. If we get 
a $1$ the number of monomials can either stay the same or decrease, but it cannot increase.
So, as $f$ starts with $r$ monomials, we get that the maximal number of $0$s queried before 
all monomials are eliminated is at most $h \cdot \log(r)$. Hence 
$\PzeroDT[f] \le h \cdot \log(r) = \poly(\log r)$, as claimed.

Thus, from now on we focus on proving \Cref{theorem:sparsity-to-hitting}. Let $f$ be a 
boolean function of sparsity $r$, and let $\mon[f]$ denote the set system of its monomials.
We consider four complexity measures associated with it:
\begin{enumerate}
    \item The hitting set complexity ($\HSC$) is the minimal size of a hitting set for it. 
        This is what we are trying to bound, and can be phrased as an covering integer program.
    
    \item The fractional hitting set complexity ($\FHSC$) is the fractional relaxation for 
        $\HSC$. Here, we want a distribution over variables that hits every monomial with high 
        probability, which can be phrased as a fractional covering linear program.
    
    \item The fractional monotone block sensitivity ($\FMBS$) is the dual linear program. The 
        reason for the name would become clear soon. It can be phrased as a fractional 
        packing linear program.
    
    \item The monotone block sensitivity ($\MBS$) is the integral version of $\FMBS$. It equals 
        the maximal number of pairwise disjoint monomials in $f$. Equivalently, it is block
        sensitivity of $f$ at $0^n$. It can be phrased as a packing integer program.  
\end{enumerate}

More generally, given $s \subseteq [n]$, let $f_s$ denote the restriction of $f$ given by 
setting $x_i = 1$ for all $i \in s$. It will be convenient to identify $s$ with its indicator
vector $1_s \in \bool^n$. Thus, for $z \in \bool^n$, we denote by $f_z$ the restriction of 
$f$ to the $1$s in $z$. Define $\HSC[f, z]$, $\FHSC[f, z]$, $\FMBS[f, z]$, $\MBS[f, z]$ to be the
above four measures for the monomials of $f_z$. It is simple to observe 
(see \Cref{claim:FMBS_eq_FHSC}) that for each $z$ we have:
\[
    \MBS[f, z] \le \FMBS[f, z] = \FHSC[f, z] \le \HSC[f, z].
\]
As a first step, we use existing techniques in boolean function analysis techniques to bound 
$\MBS[f, z]$ in terms of the sparsity of $f$. We show in \Cref{lemma:MBS_sparsity} that 
\[
    \MBS[f, z] \le
        O((\log \spar[f_z])^2) \le 
        O((\log r)^2).
\]
Thus, to complete the picture, we would like to show that if $\MBS[f, z]$ is low then so is 
$\HSC[f, z]$. This however is false, if one compares them point wise (for a single $z$). 
However, we show that the measures are equivalent (up to polynomial factors) if instead 
we consider their maximal value over all $z$. Define
\[
    \MBS[f] = \max_{z \in \bool^n} \MBS[f, z]
\]
and similarly define $\FMBS[f],\FHSC[f],\HSC[f]$.
We show in \Cref{lemma:FMBS_MBS} that 
\[
    \FMBS[f] = O(\MBS[f]^2),
\]
linear programming duality gives $\FHSC(f) = \FMBS(f)$, and we show in \Cref{lemma:HS_sparsity}
that 
\[
    \HSC(f) = O(\FHSC(f) \cdot \log r).
\]
This completes the proof of \Cref{theorem:sparsity-to-hitting}.

We also briefly discuss \Cref{thm:lifting}. The improved exponent is obtained by using the bound 
$\MBS(f) = O(\PCC[\andFunction{f}])$, which we prove in \Cref{corollary:PCC_MBS}. Its proof is based
on the observation that if $\MBS(f)=b$ then $\andFunction{f}$ embeds as a sub-function unique 
disjointness on $b$ bits, and combine it with known lower bounds on the communication complexity of
unique disjointness.

\subsection{Generalizations}
Several of our definitions and techniques readily extend to non-boolean functions,
namely to functions $f : \bool^n \to \R$. We refer the reader to \Cref{sec:prelim}
for the relevant definitions and \Cref{sec:generalization} for a detailed discussion
of the generalized results. Here, we briefly state some of the results.

\begin{restatable}{theorem}{generalhsc}
\label{thm:gen_hsc}
    Let $f : \bool^n \to \R$ be a multlinear polynomial
    with sparsity $r$. Suppose $\MBS[f] = m$. Then the hitting set complexity of $f$ is
    bounded by
    \begin{equation*}
        \HSC[f] = O(m^2\log r). 
    \end{equation*}
\end{restatable} 


\begin{restatable}{theorem}{maingeneral}
\label{thm:maingeneral}
    Let $f: \bool^n \to S$ for $S \subset \R$. Assume that $\spar[f] = r$ and
    $|S| = s$. Then the hitting set complexity of $f$ is bounded by
    \begin{equation*}
        \HSC[f] = O(s^4 (\log r)^5).
    \end{equation*}
\end{restatable}

\begin{restatable}{theorem}{generalsetsystem}
\label{thm:generalsetsystem}
    Let $\mathcal{F}=\{S_1,\cdots,S_r\}$ be a set system. Then for any 
    $m \ge 1$, at least one of the following holds:
    \begin{enumerate}
        \item $\mathcal{F}$ has a hitting set of size $h=O(m^2\log r)$. 
        \item There exists a subset $T\subset [n]$ so that 
            $\mathcal{F}_T = \{S_1\setminus T,\cdots,S_r\setminus T\}$ 
            contains $m$ pairwise disjoint sets.
    \end{enumerate}
\end{restatable}

\paragraph*{Acknowledgements.} S.L. thanks Kaave Hosseini, who was involved
in early stages of this work. S.M. thanks Russell Impagliazzo for useful 
discussions throughout the course of this work.

\section{Preliminaries}
\label{sec:prelim}

This section introduces a number of complexity measures used in the proofs
of our main results. We start by collecting some simple definitions, proceed
to define the complexity measures, and then provide some examples which 
clarify some aspects of these definitions. 

Throughout this section, fix a boolean function $f : \bool^n \to \R$. 
We identify elements of $\bool^n$ with subsets of $[n]$. Namely, we identify
$z \in \bool^n$ with the set $\set{i \ :\  z_i=1}$, and shorthand 
$[n] \setminus z = \set{i \ :\  z_i = 0}$.
Given two inputs $z, w \in \bool^n$ we denote by $z \lor w$ their union and by
$z \land w$ their intersection. The partial order on $\bool^n$ is defined by
the relation $z \leq w$, satisfied precisely when $z$ is a subset of $w$.
Define $f_z : \bool^{[n] \setminus z} \to \R$ to be the restriction of $f$ to 
inputs which are consistent with the $1$s in $z$; namely $f_z(w) = f(z \lor w)$.
Define 
$\cW(f, z) = \set{ w \in \bool^{[n] \setminus z} \ :\  f(z) \neq f(z \lor w) }$ 
and note that it can be equivalently defined as 
$\cW(f, z) = \set{ w \in \bool^{[n] \setminus z} \ :\  f_z(w) \neq f_z(0) }$.

Recall also the notation from the proof overview.
Any $f : \bool^n \to \R$
can be written uniquely
as a multilinear real-valued polynomial 
$f(x) = \sum_{s \subseteq [n]} \alpha_s \prod_{i \in s} x_i$. The \emph{sparsity}
of $f$, denoted $\spar(f)$, is the number of nonzero coefficients in the 
polynomial expansion of $f$. Next, let 
$\mon[f] = \set{s \subseteq [n] \ :\  \alpha_s \neq 0,\ s \neq 0^n}$ denote the
set system of non-zero, non-constant monomials in $f$ when written as a multilinear
polynomial. We emphasize that the coefficient $\alpha_{\emptyset}$ is not included 
in $\mon[f]$; $\alpha_{\emptyset}$ is inessential, since we are interested in hitting
sets for monomials and $\emptyset$ is trivially hit by any set. Observe that 
$|\mon[f]| \in \set{\spar[f], \spar[f] - 1}$.

For any set system $\cF$ over $[n]$, an element $z \in \cF$
is \textit{minimal} if there does not exist $w \in \cF$ with $w<z$.

\begin{claim}
    \label{claim:minimal_monomials}
    Fix $f: \bool^n \to \R$, $z \in \bool^n$
    and $\cW(f, z)$, $\mon[f_z]$ as above. Then, for
    any $w \in \bool^n$, $w$ is a minimal
    element in $\cW(f, z)$ if and only if $w$ is a minimal
    element in $\mon[f_z]$. 
\end{claim}

\begin{proof}
    We assume for simplicity that $z = \emptyset$ so that $f_z(w) = f(w)$, 
    $f(\emptyset) = \alpha_\emptyset$
    and write $\cW = \cW(f, \emptyset)$.
    Suppose $w \in \mon[f]$ is a minimal element. Writing
    $f$ as a multilinear polynomial, we get
    $f(w) = \sum_{u \leq w} \alpha_u$.
    Since $\alpha_w$ is minimal, $f(w) = \alpha_\emptyset + \alpha_u$
    and so $f(w) \neq f(\emptyset)$ and $w \in \cW$. Additionally, $w$ is minimal
    in $\cW$ because if $w'<w$ then the non-constant terms of 
    $f(w') = \sum_{u \leq w'} \alpha_u$ are all $0$, hence $f(w')=f(0)$ and 
    $w' \not\in \cW$. 
    
    In the other direction, suppose $w \in \cW$ is a minimal 
    element. Assume there is $w'<w$ in $\mon[f]$; choosing such a minimal $w'$,
    we would get $f(w') \ne f(0)$ which violates the minimality of $w$. Similarly,
    if $w \not\in \mon[f]$ then we get $f(w) = \sum_{u \leq w} \alpha_u = f(0)$,
    which violates the assumption that $w \in \cW$. Thus $w$ is a minimal element
    in $\mon[f]$.
\end{proof}

\subsection{Monotone block sensitivity}

First, we consider \textit{monotone
block sensitivity}, a variant of the standard notion of \textit{block sensitivity}
due to Nisan and Szegedy~\cite{nisan1994degree}. In a nutshell, this is a 
``directed'' restriction of block sensitivity, where we can only change an
input by flipping $0$'s to $1$'s. We also define MBS (and all other complexity 
measures introduced later in this section) with respect to real-valued functions
over $\bool^n$. This differs from block sensitivity, which
is usually (though not always) studied in the context of boolean-valued functions.
The generalization to real-valued $f$ will be immaterial to some of our proofs, 
permitting us to draw more general conclusions regarding the monomial structure 
of multilinear polynomials; see \Cref{sec:generalization} for more details. 

Say that two inputs $z,w$ are disjoint if $z \land w = 0^n$;
namely, their corresponding sets are disjoint. 

\begin{definition}[Monotone block sensitivity]
    For $f: \bool^n \to \R$ and $z \in \bool^n$, 
    the \textit{monotone block sensitivity of $f$ at $z$}, 
    denoted $\MBS(f, z)$, is the largest  integer $k$ 
    such that there  exist $k$ pairwise disjoint inputs 
    $w_1, \dots, w_k \in \cW(f, z)$. We denote $\MBS[f] = \max_z \MBS[f, z]$. 
\end{definition}

For two motivating examples, observe that for the $n$-bit AND and OR functions
we have $\MBS[\AND] = 1$ and $\MBS[\OR] = n$, respectively. 

\begin{remark}
    We emphasize that $\cW(f, z) \subseteq \bool^{[n] \setminus z}$, so each $w_i$
    is disjoint from $z$. This corresponds to the standard definition of block 
    sensitivity where we restrict each block $w_i$ to be disjoint from the support
    of $z$.
\end{remark}

\begin{remark}
    Suppose $w_1, \dots, w_k$ are minimal witnesses that $\MBS[f, z] = k$
    in the sense that for any $i \in [k]$ there is no $w'_i < w_i$
    so that $w'_i \in \cW(f, z)$. Then by \Cref{claim:minimal_monomials},
    each $w_i$ is a minimal element in $\mon[f_z]$.
\end{remark}

As alluded to in the proof overview, $\MBS[f, z]$
can be phrased as the value of a particular
set packing linear program (LP). Fixing $z$,
write $\cW = \cW(f, z)$. The program optimizes
over variables $a_w$ for each $w \in \cW$.

\begin{align*}
    \text{maximize } & \sum_{w \in \cW} a_w \\
    \text{subject to } & \sum_{w \in \cW: w_i = 1} a_w \le 1 \quad \text{ for all } i \in [n]\\
        & a_w \in \bool \quad \text{for all } w \in \cW
\end{align*}

\emph{Fractional monotone block sensitivity} ($\FMBS$) is obtained
by relaxing the constraints in the above LP, allowing variables
$a_w$ to assume non-integral values in $[0, 1]$. We use an alternative
formulation of $\FMBS$ whose equivalence to the LP formulation is simple to verify. 

\begin{definition}[Smooth distribution]
    A distribution $\D$ over $\bool^n$ is said to be $p$-smooth if for any 
    $i \in [n]$ it holds that $\Pr_{w \sim \D}[w_i = 1] \le p$.
\end{definition}

\begin{definition}[Fractional monotone block sensitivity]
\label{def:FMBS}
    The fractional monotone block sensitivity of a function 
    $f : \bool^n \to \R$ at an input $z \in \bool^n$, denoted
    $\FMBS(f,z)$, 
    is equal to $1/p$, where $p > 0$ is the smallest number for which 
    there exists a $p$-smooth distribution $\D$ supported on a subset of
    $\cW(f, z)$. We denote $\FMBS[f] = \max_z \FMBS[f, z]$.
\end{definition}

\begin{remark}
    To see the equivalence between this definition of $\FMBS$ and 
    the LP formulation, notice that a solution 
    $\set{a_w \ :\  w \in \cW(f,z)}$ to the LP with $s = \sum a_w$ gives
    rise to a $1/s$-smooth distribution $\D$ over $\cW(f,z)$ via
    $\D(w) = a_w / s$. 
\end{remark}

\begin{remark}
    Clearly, any solution to the fractional program for $\FMBS$ 
    is a solution to the integral program for $\MBS$. Hence, both being 
    \emph{maximization} problems, $\FMBS$ upper bounds $\MBS$. Later, we 
    prove in \Cref{lemma:FMBS_MBS} that the converse of this inequality 
    holds in the sense that $\FMBS[f]$ is upper bounded by a polynomial 
    in $\MBS[f]$.
\end{remark}

\begin{remark}
    Fractional block sensitivity (the non-monotone variant) was considered
    by Tal in~\cite{tal2013properties}. Tal mentions explicitly the problem
    of finding separations between fractional block sensitivity and sensitivity. 
\end{remark}

\subsection{Hitting set complexity}

Next, we consider \emph{hitting set complexity}. This can be viewed
as a variant of \emph{certificate complexity},
a commonly-studied quantity in standard query complexity. 

\begin{definition}[Hitting set complexity]
    The hitting set complexity of a function $f : \bool^n \to \R$
    at an input $z \in \bool^n$, denoted $\HSC[f, z]$, is the minimal
    size of a set $H \subseteq [n]$ which intersects all sets in $\cW(f,z)$. 
    In other words, for every $w \in \cW(f, z)$ there is some 
    $i \in H$ so that $w_i = 1$. 
    We denote $\HSC[f] = \max_z \HSC[f, z]$. 
\end{definition}

Similarly to $\MBS$, it is simple to see that the $n$-bit AND and OR functions
have $\HSC[\AND_n] = 1$ and $\HSC[\OR_n] = n$, respectively. 

\begin{remark}
    It suffices to consider $H \subseteq [n]$ which have non-empty
    intersection with any \emph{minimal} element of $\cW(f, z)$.
    This is simply because if $H$ hits an element $w$
    then it also hits every superset of $w$. 
\end{remark}

\begin{remark}
    Suppose $H \subseteq [n]$ with $|H| = b$ witnesses $\HSC[f, 0^n] = b$.
    By the previous remark and \Cref{claim:minimal_monomials}, one can see
    that $H$ is hitting set of $\mon[f]$. 
\end{remark}

We can also phrase $\HSC[f, z]$
as the value of a certain set covering LP. Putting
$\cW = \cW(f, z)$, 
the LP optimizes over the variable $\set{b_i \ :\  i \in [n]}$ as follows:

\begin{align*}
    \text{minimize} &\ \sum_{i \in [n]} b_i \\
    \text{subject to} &\ \sum_{i \in [n]: w_i = 1} b_i \ge 1 \quad \text{ for all } w \in \cW\\
        &\ b_i \in \bool \quad \text{ for all } i \in [n]
\end{align*}

One can easily verify that this LP is dual to the 
LP defining monotone block sensitivity.
Fractional hitting set complexity is obtained from hitting
set complexity by relaxing each constraint $b_i \in \bool$
to $b_i \in [0, 1]$. We give an alternative definition, equivalent
to the LP formulation:

\begin{definition}[Fractional hitting set complexity]
\label{def:FHSC}
    The fractional hitting set complexity of a function 
    $f : \bool^n \to \R$ at an input $z \in \bool^n$, 
    denoted $\FHSC[f, z]$, is $1/p$, where $p > 0$ is the smallest
    number for which there
    exists a distribution $\D$ of indices $i \in [n]$ with the
    property that $\Pr_{i \sim \D}[w_i =1 ] \ge p$ 
    for each $w \in \cW(f, z)$.
    We denote $\FHSC[f] = \max_z \FHSC[f, z]$.
\end{definition}

\begin{remark}
    The same reasoning as the $\FMBS$ case can be used to show that this definition is equivalent to the LP definition. 
    Also by analogous reasoning, $\FHSC[f, z] \leq \HSC[f, z]$ (recalling that $\FHSC$ is a minimization problem). 
\end{remark}

The LPs defining $\FHSC$ and $\FMBS$ are dual, so 
linear programming duality yields $\FHSC[f, z] = \FMBS[f, z]$. Combined
with the remarked-upon relationships between $\MBS / \FMBS$ and $\HSC / \FHSC$, 
we conclude the following:
\begin{claim}
\label{claim:FMBS_eq_FHSC}
    For any function $f: \bool^n \to \R$ and input $z \in \bool^n$, 
    \[
        \MBS[f, z] \le \FMBS[f,z] = \FHSC[f, z] \le \HSC[f, z].
    \]
\end{claim}

\subsection{Some informative examples}

To digest the definitions, some examples are in order.
We start by noting that there are large gaps in the inequalities
from \Cref{claim:FMBS_eq_FHSC} for \textit{fixed} $z$. These correspond
to integrality gaps for the set cover and hitting set linear programs (of
which $\FHSC[f, z]$ and $\FMBS[f, z]$ are a special case), 
which are central to combinatorial optimization. 

The first example gives a separationbetween $\FMBS[f, z]$ and $\MBS[f, z]$.

\begin{example}[Projective plane]
    \label{example:projective_plane}
    For a prime power $m$, let $P$ be the set of 1-dimensional
    subspaces of $\F_m^3$ and $L$ the set of 2-dimensional
    subspaces of $\F_m^3$. $P$ is the set of \emph{points}
    and $L$ is the set of \emph{lines}. 
    Note that $|P| = |L| = m^2 + m + 1$.
    
    It is well-known that $P$ and $L$ form a \emph{projective plane},
    in that they enjoy the following relationship:
    \begin{enumerate}
        \item Any two points in $P$ are contained in exactly one 
            line in $L$. Moreover, each point is contained in $m + 1$ lines.
        \item Any two lines in $L$ intersect at exactly one point in $P$.
            Moreover, each line contains $m + 1$ points.
        \item There are 4 points, no 3 of which lie on the same line.
    \end{enumerate}
    For more background on finite geometry, 
    see, for example, \cite{ball2011introduction}.
    
    Let $n = m^2 + m + 1$, thinking of each $i \in [n]$ as corresponding
    to a point $p_i \in P$. For lines $\ell \in L$, let
    $S_\ell = \set{ i \in [n] \ :\  p_i \text{ is contained in } \ell}$ be the
    set of (indices of) points incident to $\ell$ and define $f: \bool^n \to \bool$ as 
    \[
            f(z) = \biglor_{\ell \in L} \Big(\bigland_{i \in S_\ell} z_i\Big).
    \]
    
    Since any two lines intersect at a point, any $\ell_1, \ell_2 \in L$
    have $S_{\ell_1} \cap S_{\ell_2} \neq \emptyset$. This implies 
    $\MBS(f, 0^n) = 1$. On the other hand, because each line contains $m+1$ points, 
    $\Pr_{i \in [n]}[i \in S_\ell] = (m+1)/(m^2 + m + 1)$ when $i$ is uniform and
    therefore $\FMBS(f, 0^n) \approx m \approx \sqrt{n}$. 
\end{example}

The next example gives a similar separation
between $\FHSC(f, z)$ and $\HSC(f, z)$.

\begin{example}[Majority]
    \label{example:maj}
    For $n$ even, let $f(z) = \mathbf{1}[\sum_i z_i \ge n/2]$ be the Majority function. 
    The minimal elements of $\mon[f]$ consist of sets $s$ with $n/2$ members. 
    
    Any set $s$ of size at most $n/2$ will fail to hit $[n] \setminus s \in \mon[f]$.
    Therefore any hitting set for the monomials of $f$, namely for $\mon[f]$, has size more than
    $n/2$. In particular, $\HSC(f, 0^n) = n / 2 + 1$ (clearly $n/2 + 1$ suffices). 
    On the other hand, the uniform distribution over $[n]$ satisfies
    $\Pr[i \in s] = 1/2$ for any minimal monomial $s \in \mon[f]$. Hence
    $\FHSC[f, 0^n] = 2$. 
\end{example}

\noindent These two examples show that it will be necessary to utilize
the fact that $\MBS$ and $\HSC$ are defined as the maximum over
all inputs.

The next examples shows that $\HSC[f]$ and $\MBS[f]$ can be constant while
$\spar(f)$ grows exponentially. 

\begin{example}[AND-OR]
    \label{example:and-or}
    Consider a string $z \in \bool^{2n}$
    written as $z = xy$ for $x,y \in \bool^n$.
    Define 
    \[
        f(x,y) = \bigland_{j \in [n]} \left(x_j \lor y_j\right).
    \]
    One can verify that $\HSC[f] = \MBS[f] = 2$.
    On the other hand, writing $f$ as a multilinear polynomial
    yields
    \[
        f(x,y) = \prod_{j \in [n]} (x_i + y_j - x_j \cdot y_j),
    \]
    which clearly has sparsity exponential in $n$. 
\end{example}

\noindent Note that this holds for the \textit{global} (i.e. maximizing
over $\bool^n$) definitions of $\MBS$ and $\HSC$. To see the 
significance of this example,  recall from the proof overview that we 
are interested in eventually showing $\MBS(f) \leq O((\log \spar(f))^2)$.
This example shows that this latter inequality can be very far from the
truth; we are able to make up for this discrepancy by using the 
low-sparsity assumption multiple times. 

Finally, we include an example which will become relevant to our 
applications to communication complexity in \Cref{sec:communication}.

\begin{example}[Redundant indexing]
    \label{example:redundant-indexing}
    Let $k \ge 1$, and  
    consider two sets of variables $\set{x_S}_{S \subseteq [k]}$ and 
    $\set{y_i}_{i \in [k]}$ of sizes $2^k$ and $k$, respectively. 
    Let $n = 2^k + k$ and define
    \[
        f(x, y) = \sum_{i \in [k]} 
            \left(\prod_{S \ :\  i \in S} x_S\right) 
                (1 - y_i) \left(\prod_{j \neq i} y_j\right).
    \]
    In words, $f(x, y) = 1$ when $y$ has weight exactly $k-1$ with $y_i = 0$
    and $x_S = 1$ for every $S$ containing $i$.

    
    By the mutilinear representation, one can see that
    the sparsity of $f$ is $2k \sim \log n$. 
    Moreover, $\HSC(f) \le 2$. To see why, 
    consider an input $z = (a, b)$ and note that $f$ restricted
    to inputs $w = (x, y) \geq z$ becomes
    \[
        f'(x, y) = \sum_{i : b_i = 0} 
            \left(\prod_{S: i \in S, a_S = 0} x_S\right) 
                (1 - y_i)\left(\prod_{j: j \neq i, b_j = 0} y_j\right).
    \]
    In particular, if $a \ne 1^{[2^k]}$ then the variable $x_{[2^k] \setminus a}$ 
    hits all the monomials, and if $a=1^{[2^k]}$ then any two $y_i,y_j$ hit all 
    the monomials.
    
\end{example}

We view this as an important example in understanding the $\log n$
factor currently present in the statements of \Cref{thm:logrank}
and \Cref{thm:lifting}. This connection will be discussed in 
more detail in \Cref{sec:discussion}.

\section{Proof of \texorpdfstring{\Cref{theorem:sparsity-to-hitting}}{Theorem \ref{theorem:sparsity-to-hitting}}}


We recall the statement of \Cref{theorem:sparsity-to-hitting}.

\mainthm*

The proof relies on three lemmas which provide various relationships
between $\spar[f]$, $\MBS[f]$ and $\HSC[f]$, as well as their
fractional variants. In this subsection, we will state the lemmas and
show how \Cref{theorem:sparsity-to-hitting} follows as a consequence. 
Then, in the following subsections, we prove the lemmas. 

The first gives an upper bound on the monotone block sensitivity
of a \textit{boolean-valued} $f$ in terms of its sparsity. 

\begin{restatable}{lemma}{lemmaone}
    \label{lemma:MBS_sparsity}
    For any boolean function $f : \bool^n \to \bool$,
    $\MBS[f] = O(\log(\spar[f])^2)$.
\end{restatable}

\noindent We stress that this only holds for boolean-valued functions. 
To some extent, we will be able to relax this condition when we
consider generalizations in \Cref{sec:generalization}. Additionally,
we note that this inequality can be very from tight: \Cref{example:and-or}
 gives a function with constant MBS but exponential sparsity. 

The second lemma shows that $\FMBS$ and $\MBS$ are equivalent
up to a polynomial factor.
Unlike \Cref{lemma:MBS_sparsity}, this holds for any real-valued 
function. 

\begin{restatable}{lemma}{lemmatwo}
    \label{lemma:FMBS_MBS}
    For any function $f:\bool^n \to \R$, $\FMBS[f] = O(\MBS[f]^2)$. 
\end{restatable}

The third lemma, which also holds for any real-valued function,
upper bounds the hitting set complexity of $f$ in terms of $\FMBS(f)$
and $\spar[f]$.

\begin{restatable}{lemma}{lemmathree}
    \label{lemma:HS_sparsity}
    For any function $f : \bool^n \to \R$, 
    $\HSC[f] \le \FMBS[f] \cdot \log(\spar[f])$.
\end{restatable}

\Cref{theorem:sparsity-to-hitting} now follows quite readily
from the three lemmas.

\begin{proof}[Proof of \Cref{theorem:sparsity-to-hitting}]
    Fix a boolean function $f$ with sparsity $r$ as in the theorem statement.
    By \Cref{lemma:MBS_sparsity}, $\MBS[f] = O((\log r)^2)$.
    By \Cref{lemma:FMBS_MBS}, $\FMBS[f] = O((\log r)^4)$.
    Finally, by \Cref{lemma:HS_sparsity}, $\HSC[f] \leq O((\log r)^5)$,
    as desired. 
\end{proof}

\subsection{MBS from sparsity}

We begin by proving \Cref{lemma:MBS_sparsity}.

\lemmaone*

The proof uses a
well-known relationship between the degree and the sensitivity
of boolean functions \cite{nisan1994degree}. The \textit{sensitivity} $S(f)$ of a 
boolean function $f$ is the largest $s$ so that there exists
an input $z$ and $s$ coordinates $\set{i_1, \dots, i_s}$ so that
$f(z) \neq f(z \lxor e_{i_j})$ for all $j \in [s]$.

\begin{claim}[Nisan-Szegedy, \cite{nisan1994degree}]
    For any boolean function $f: \bool^n \to \bool$, $S(f) = O(\deg(f)^2)$.
\end{claim}

\begin{proof}[Proof of \Cref{lemma:MBS_sparsity}]
    Suppose $\MBS[f] = k$, witnessed by pairwise disjoint 
    $z, w_1, \dots, w_k \subseteq [n]$. Namely, 
    $f(z) \ne f(z \lor w_i)$ for $i \in [k]$. 
    Let $g : \bool^k \to \bool$ denote the function obtained from 
    $f$ by identifying variables in each $w_i$ and setting all 
    variables not occurring in any $w_i$ to the corresponding 
    bit in $z$. That is,  $g(x) = f(z + \sum x_i w_i)$. Note 
    that $S(g) = k$, since $g(0) \ne g(e_i)$ for $i \in [k]$, 
    and $\spar[g] \le \spar[f]$. 

    Let $r = \spar[f]$. We will reduce the degree of $g$ to 
    $d = O(\log r)$ by repeating the following process $k / 2$ times:
    set to zero the coordinate which appears in the largest number
    of monomials of degree at least $d$.
    
    Let $M_i$ denote the number of monomials of degree at least $d$ 
    remaining after the $i$-th step. Initially $M_0 \le r$. Next,
    note that if $M_i > 0$, then there is a variable that occurs in
    at least a $d/k$ fraction of the monomials of degree $\ge d$. 
    We therefore obtain the recurrence $M_{i + 1} \le (1 - d / k) M_i$.
    After $k / 2$ steps, 
    $M_{k / 2} \le (1 - d / k)^{k / 2} r \le \exp(-d / 2) r < 1$ 
    for $d = O(\log r)$. As $M_{k / 2}$ is an integer, we obtain that
    $M_{k / 2}$ is zero. 
    
    Let $h$ denote the function obtained by this restriction 
    process. Since $M_{k / 2} = 0$ we have $\deg(h) < d$. Moreover, since
    $g$ had full sensitivity at $0^k$ and we restricted only $k / 2$ 
    coordinates, $S(h) \ge k/2$. Finishing up, we have 
    $k/2 \le S(h) = O(\deg(h)^2) = O((\log r)^2)$, completing 
    the proof. 
\end{proof}

\subsection{Fractional vs. integral solutions for MBS}

This subsection proves \Cref{lemma:FMBS_MBS}, restated here:

\lemmatwo*

We first need the following claim, which states
that any function $f: \bool^n \to \R$
is not too sensitive to noise which 
is $q$-smooth for $q \ll 1/\FMBS[f]$.

\begin{claim}\label{claim:bias_restrict}
    Let $f:\bool^n \to \R$, $z \in \bool^n$ and $\D$ a distribution on 
    $\bool^{[n] \setminus z}$. Assume that $\D$ is $q$-smooth for some 
    $q \in (0, 1]$. Then
    \[
        \Pr_{w \sim D}[f(z) \neq f(z \lor w)] \le q \cdot \FMBS[f, z]. 
    \]
\end{claim}

\begin{proof}
    Assume $\FMBS[f, z] = 1 / p$. We may assume $q < p$ as otherwise the 
    claim is trivial. Let $\delta = \Pr_{w \sim D}[f(z) \neq f(z \lor w)]$.
    Let $\D'$ be the distribution $\D$ restricted to inputs $w$ such that 
    $f(z) \neq f(z \lor w)$. Observe that $\D'$ is $(q / \delta)$-smooth, 
    and is supported on inputs $w$ such that $f(z) \neq f(z \lor w)$. As 
    $\FMBS(f,z) = 1 / p$ we have $q / \delta \ge p$ which implies the claim.
\end{proof}

    

\begin{proof}[Proof of \Cref{lemma:FMBS_MBS}]
    Let $\FMBS(f) = 1/p$. Let $z \in \bool^n$ such that $\FMBS(f,z) = 1/p$, and let $\D$ be a $p$-biased distribution supported on $\cW(f,z)$. 
    
    Fix $k$ to be determined later, and sample inputs 
    $w_1, \dots,w_k \sim \D$ independently. Let $u$ denote all the elements
    that appear at least in two of the $w_i$, namely
    \[
        u = \biglor_{i \neq j} \left( w_i \bigland w_j \right).
    \]
    The main observation is that $u$ is $q$-biased for $q=(pk)^2$. This holds
    since for every $\ell \in [n]$ we have
    \[
        \Pr[u_\ell=1] \le 
            \sum_{i \neq j} \Pr[(w_i)_\ell = 1, (w_j)_\ell = 1] \le 
            k^2 p^2. 
    \]
    Define the following ``bad'' events:
    \[
        E_0 = [f(z) \neq f(z \lor u)], \quad E_t = [f(z \lor w_t) \neq 
        f(z \lor w_t \lor u)]  \text{ for } t \in [k].
    \]
    We claim that $\Pr[E_t] \le q/p = p k^2$ for all $t=0, \dots,k$. The proof
    for $E_0$ follows directly from \Cref{claim:bias_restrict}. To see why it
    holds for $E_t$ for $t=1, \dots,k$, define $u_t$ to be the elements that
    appear in two sets $w_i$, excluding $w_t$, namely
    \[
        u_t = \biglor_{i \neq j, \; i,j \neq t} \left( w_i \bigland w_j \right).
    \]
    Observe that $w_t,u_t$ are independent, that $u_t$ is $(pk)^2$-biased and 
    that $w_t \vee u = w_t \vee u_t$. Thus \Cref{claim:bias_restrict} gives that,
    for any fixing of $w_t$, we have
    \[
        \Pr_{u_t}[f(z \vee w_t) \neq f(z \vee w_t \vee u_t) \; | \; w_t] \le 
        q \cdot \FMBS(f,z \vee w_t) \le q \cdot \FMBS(f) = q / p = pk^2.
    \]
    The claim for $E_t$ follows by averaging over $w_t$.
    
    
    Pick $k = 1/(2\sqrt{p})$, meaning 
    $E_t$ occurs with probability at most $1/4$
    for each $0 \leq t \leq k$. Then conditioning on
    $\neg E_0$ will increase the probability
    of any event by a factor of at most $1/(1 - 1/4) = 4/3$. 
    In particular, because $\Pr[E_t] \leq pk^2 = 1/4$ for any $t$,
    we have $\Pr[E_t | \neg E_0] \leq 1/3$ for any $t \neq 0$. 
    This means that we can sample the $w_t$'s conditioned on 
    $\neg E_0$, and still
    be sure that every $\neg E_t$ occurs with probability at least $2/3$.
    Averaging, some setting of the $\{w_t\}$ satisfies $\neg E_0$
    and at least $2/3$ of  $\neg E_t$ for $1 \leq t \leq k$. Fix these $\{w_t\}$. 
    
    Define $z' = z \vee u$ and $w'_t = w_t \setminus u$. 
    For every $1 \leq t \leq k$ for which $\neg E_t$ holds, we have 
    \[
         f(z') = f(z), \qquad 
         f(z' \vee w'_t) = f(z \vee w_t).
    \]
    Thus $f(z') \neq f(z' \vee w'_t)$ for at least
    $2k/3$ choices of $w'_t$.
    Moreover, $z',w'_1, \dots, w'_k$ are pairwise disjoint. 
    Hence $\MBS(f) \geq 2k/3$. This completes the proof, 
    by recalling that $k = 1/(2\sqrt{p})$ with $\FMBS(f) = 1/p$. 
\end{proof}

A notable feature of this proof is that we need to employ upper 
bounds on the fractional block sensitivity for more than one choice 
of input. This is actually necessary; there is a function $f$ based 
on the projective plane for
which $\MBS(f, z) = 1$ and $\FMBS(f, z) \sim \sqrt{n}$ at a point $z$. See 
\Cref{example:projective_plane} for details. 

\subsection{Hitting sets from sparsity}

Our final lemma is an upper bound on the hitting set complexity of
any $f: \bool^n \to \R$ in terms of $\FMBS(f)$ and $\log(\spar[f])$.
Recall that $\FMBS$ and $\FHSC$ are equal, so such an upper 
bound implies that $\FHSC$ and $\HSC$ are polynomially related for
sparse boolean functions. 

\lemmathree*

Before proving it, we need two straightforward
claims which we will use again later on. 
The first allows us to find (non-uniformly) indices $i \in [n]$ 
which hit a large fraction of $\mon[f]$, given
that $f$ has small $\FMBS / \FHSC$ at $0^n$. 

\begin{claim}
    \label{claim:transfer}
    Suppose $\FMBS[f, 0^n] = \FHSC[f, 0^n] = k$ and this is witnessed by a 
    distribution $\D$ over $[n]$. Then 
    \begin{enumerate}
        \item $\Pr_{i \sim \D}[i \in w] \geq 1/k$
            for every $w \in \mon[f]$. That is, $\D$ is also a fractional
            hitting set for the \emph{monomials} of $f$. 
        \item There is some $i$ in the support of $\D$ which hits
        a $1 / k$-fraction of $\mon[f]$.
    \end{enumerate}
\end{claim}

\begin{proof}
    Note that the second part of the claim follows from the first
    by an averaging argument, so we are contented to prove the
    first part of the claim. 

    Let $\D$, $\FHSC(f, 0) = k$ be as stated, 
    so that $\Pr_{i \sim \D}[i \in w] \geq 1/k$
    for all $w \in \cW(f, z)$. By \Cref{claim:minimal_monomials},
    it is the case that $\Pr_{i \sim \D}[i \in w] \geq 1/k$
    for any minimal monomial $w$. The measure of $\D$ on some $w$ is 
    non-decreasing with respect to taking supersets, meaning
    $\Pr_{i \sim D}[i \in w] \geq 1/k$ for every monomial $w \in \mon[f]$. 
\end{proof}

The second claim says that $\FHSC(f)$ is non-increasing
under restrictions. For simplicity, we only consider
reductions which set a single bit (which can be extended
to more bits by induction). 

\begin{claim}
    \label{claim:FHSC-is-nonincreasing}
    Let $f : \bool^n \to \R$ be a function, $i \in [n]$ and $b \in \bool$.
    Let $f':\bool^{[n] \setminus \set{i}} \to \R$ be the function obtained
    by restricting to inputs with $x_i = b$. Then
    \[
        \FHSC[f'] \le \FHSC[f].
    \]
\end{claim}

\begin{proof}
    Fix $z \in \bool^{[n] \setminus \set{i}}$. We will show that 
    $\FHSC[f', z] \le \FHSC[f, z^*]$ where $z^* = z$ if $b = 0$ and 
    $z^* = z \cup \set{i}$ if $b = 1$. In either case,  
    $\FHSC[f',z] \le \FHSC(f)$ and hence $\FHSC(f') \le \FHSC(f)$.

    Consider first the case of $b=0$, and assume that 
    $\FHSC[f, z] = 1 / p$. Recall that $f_z$ is the restriction of $f$
    to inputs $x \ge z$, and that $\FHSC[f, z] = \FHSC[f_z, 0]$.
    By definition, there is a distribution $\D$ over $[n]$ such that
    for every $w \in \mon[f_z]$ we have $\Pr_{i \sim D}[w_i = 1] \ge p$.
    Observe that $\mon[f'_z] \subset \mon[f_z]$ since setting a variable
    to $0$ can only remove monomials. Thus we get $\FHSC[f', z] \le \FHSC[f, z]$.

    Next, consider the case of $b=1$. Note that $f'_z = f_{z \cup \set{i}}$
    and hence $\FHSC[f', z] = \FHSC[f, z \cup \set{i}]$.
\end{proof}

\begin{proof}[Proof of \Cref{lemma:HS_sparsity}]
    Let $k = \FHSC(f, 0) \leq \FHSC(f)$, $S_0 = \emptyset$, $f_0 = f$ and 
    perform the following iterative process. At time $t \geq 1$, let 
    $S_t = S_{t-1} \cup \{i_t\}$ where $i_t \in [n]$ is the index which 
    hits a $1/k$-fraction of $\mon[f_{t-1}]$, guaranteed to exist by 
    \Cref{claim:transfer}. Let $f_{t} = f_{t-1}|_{z_{i_t} = 0}$. At each 
    step, the restriction $z_{i_t} = 0$ sets every monomial containing 
    $i_t$ to zero, causing the sparsity of $f_{t-1}$ to decrease by a 
    multiplicative factor $(1 - 1/k)$. Let $r_t = |\mon[f_t]|$. Since $S_t$
    is a hitting set for $\mon[f]$ when $f_t$ has no non-zero monomials, 
    this process terminates with a hitting set when
    \[
        r_t = (1 - 1/k)^t r_0 \leq e^{- t/k} r_0 < 1.
    \]
    Therefore, taking $t = k \log r_0$ suffices. 
\end{proof}

\section{Corollaries in communication complexity}
\label{sec:communication}

\subsection{Preliminaries}

Fix a \textit{boolean} function $f: \{0, 1\}^n \to \{0, 1\}$.
Let $\andFunction{f} = f \circ \wedge$ denote the AND function corresponding to $f$, given by 
$\andFunction{f}(x,y) = f(x \wedge y)$. The sparsity of $f$ characterizes the rank of $\andFunction{f}$. 
\begin{claim}
    \label{claim:rank_sparsity}
    Let $f: \bool^n \to \bool$ be a boolean function. 
    Then $\spar(f) =\rank[\andFunction{f}]$. 
\end{claim}

\begin{proof}[Proof]
    Let $f(z) = \sum_s f_s \prod_{i \in s} z_i$ be the multilinear polynomial computing $f$. 
    Then $f(x \land y)$, expanded as a multilinear polynomial, equals
    \[
        f(x \land y) = \sum_s f_s \left(\prod_{i \in s} x_i\right) \left( \prod_{i \in s} y_i\right).
    \]
    Hence we can write the $2^n \times 2^n$ communication matrix of $\andFunction{f}(x, y) = f(x \land y)$ as
    \[
        M = \sum_s f_s v_s v_s^{\top}
    \]
    where $v_s \in \bool^{2^n}$ is given by $(v_s)_x = \prod_{i \in s} x_i$. The $v_s$'s
    are linearly independent and therefore $M$ has rank equal to the number of non-zero 
    entries in the sum. 
\end{proof}

We assume familiarity with the standard notion of a \textit{decision tree}.
Our primary interest is in a variant of decision trees called 
\textit{AND decision trees}, which strengthens decision trees by allowing
queries of the conjunction of an arbitrary subset of the variables, namely 
queries of the form $\land_{i \in S} z_i$ for arbitrary $S \subseteq [n]$.
Let $\PandDT[f]$ denote the smallest depth of an AND decision tree computing 
$f$. The following simple connection to the communication complexity of 
$\andFunction{f}$ motivates our interest in this model:
\begin{claim}
\label{claim:cc_adt}
    Let $f : \bool^n \to \bool$. Then $\PCC(\andFunction{f}) \le 2 \PandDT[f]$. 
\end{claim}

\begin{proof}[Proof]
    Whenever the AND-decision tree queries a set $S \subseteq [n]$, Alice and Bob
    privately evaluate $a = \land_{i \in S} x_i$ and $b = \land_{j \in S} y_j$, exchange them
    and continue the evaluation on the sub-tree obtained by following the edge 
    labelled $a \land b$. If the decision tree height is $d$, this protocol uses 
    $2d$ bits of communication. Correctness follows from the observation that
    $\bigland_{i \in S} (x_i \land y_i)  
        = (\bigland_{i \in S} x_i) \land (\bigland_{j \in S} y_j)$.
\end{proof}

There is also a simple connection between AND-decision trees
and sparsity:

\begin{claim}
    \label{claim:adt_exp_sparsity}
    Let $f: \bool^n \to \bool$ with $d = \PandDT[f]$. Then $\spar[f] \leq 3^d$.
\end{claim}
\begin{proof}
    Assume that $f$ is computed by a depth-$d$ AND decision tree, where the first query
    is $\land_{i \in S} z_i$, and where $f_1,f_2$ are the functions computed by the left
    and right subtrees, respectively. Note that both are computed by AND decision trees
    of depth $d-1$. We have
    \[
        f(z) = \prod_{i \in S} z_i \cdot f_1(z) + \left(1 - \prod_{i \in S} z_i\right) f_2(z). 
    \]
    Thus
    \[
        \spar[f] \le \spar[f_1] + 2 \cdot \spar[f_2].
    \]
    The claim follows, since in the base case, functions computed by an AND-decision
    tree of depth $1$ has sparsity at most $2$.
\end{proof}

A related complexity measure introduced in \cite{mukhopadhyay2019lifting}, called
the \text{$0$-decision tree complexity} of $f$, is defined as follows. The 
\textit{$0$-depth} of a (standard) decision tree $\mathcal{T}$ is largest number 
of $0$-edges encountered on a root-to-leaf path in $\mathcal{T}$. The $0$-decision
tree complexity of $f$, denoted $\PzeroDT[f]$, is the  smallest $0$-depth over all
trees $\mathcal{T}$ computing $f$. The following relationship between AND decision
trees and $0$-decision tree complexity is from~\cite{mukhopadhyay2019lifting}:

\begin{claim}[\cite{mukhopadhyay2019lifting}]\label{claim:0DT_ADT}
    For any boolean function $f: \bool^n \to \bool$,
    \[
        \PzeroDT[f] \le \PandDT[f] \le \PzeroDT[f] \ceil{\log (n + 1)}.
    \]
\end{claim}

For completeness, we include the short proof.
\begin{proof}
    The first inequality follows since an AND query can be simulated by querying
    the bits in it one at a time, until the first $0$ is queried, or until they 
    are all queried to be $1$. In particular, at most a single $0$ query is made.
    This implies that an AND decision tree of depth $d$ can be simulated by a 
    standard decision tree of $0$-depth $d$.
    
    For the second inequality, let $\mathcal{T}$ be a decision tree computing $f$
    with $0$-depth $d$. Consider the subtree which is truncated after the first $0$
    is read. We can compute which leaf in the subtree is reached by doing a binary
    search on the at most $n + 1$ options, which can be implemented using 
    $\ceil{\log(n + 1)}$ computations of ANDs. Then, the same process continues on
    the tree rooted at the node reached, which has $0$-depth at most $d - 1$. 
\end{proof}

The following example shows that this gap of $\log n$ cannot be avoided. 

\begin{example}
    For $z \in \bool^n$, let $\text{ind}(z) \in [n]$ denote the
    first index $i$ for which $z_i = 0$. Let
    \[
        f(z) = 
        \begin{cases}
            1                     & \text{ if $z = 1^n$ or $z=1^{n-1} 0$} \\
            z_{\text{ind}(z) + 1} &\text{ otherwise}
        \end{cases}
    \]
    Any decision tree computing
    $f$ will have to query at most two zeroes, corresponding to
    $z_{\text{ind}(z)}$ and $x_{\text{ind}(z) + 1}$, and hence
    $\PzeroDT[f] \leq 2$. However, a direct calculation shows that
    $\spar(f) = \Omega(n)$ and therefore, by
    \Cref{claim:adt_exp_sparsity}, $\PandDT[f] = \Omega(\log n)$. 
\end{example}

We also use a lemma closely related to \Cref{lemma:HS_sparsity}.

\begin{lemma}
    \label{lemma:0DT_spars}
    Let $f: \bool^n \to \bool$ be an arbitrary boolean function.
    Then 
    \[
        \PzeroDT[f] = O( \FMBS[f] \cdot \log \spar[f] ). 
    \]
\end{lemma}

\begin{proof}
    Let $k = \FHSC[f, 0] \le \FHSC[f]$. By \Cref{claim:transfer}, 
    there is an $i \in [n]$ that belongs to at least a 
    $(1 / k)$-fraction of $\mon[f]$. Query the 
    variable $x_i$ and let $b_i \in \bool$ be the outcome.
    Let $f' : \bool^n \to \bool$ be the function $f$ restricted to 
    $x_i = b_i$. Consider the sparsity of $f'$:
    \begin{itemize}
        \item If $x_i = 0$ then $|\mon[f']| \le (1 - 1 / k) |\mon[f]|$, 
            as setting $x_i = 0$ kills a $(1 / k)$-fraction of the 
            non-constant monomials. Thus, as long as $f$ is not a 
            constant function, $|\mon[f]| \ge 1$ and we have
            \[
                \spar[f'] \le 
                \spar(f) - |\mon[f]|/k \le 
                \spar[f] (1 - 1 / 2k).
            \]

        \item If $x_i = 1$ then $\spar(f') \le \spar[f]$, since fixing 
            variables to constants cannot increase the number monomials.
    \end{itemize}
    Let $t$ the maximum number of $0$'s queried along any 
    path in the obtained decision tree. The sparsity of the
    subfunction $f'$ corresponding to a leaf must be $0$
    or else $f'$ is non-constant.  By, \Cref{claim:FHSC-is-nonincreasing} 
    $f'$ is constant when $(1 - 1/2k)^t \spar[f] \leq e^{-t / 2k} \spar[f] < 1$,
    which occurs when $t \geq 2k \cdot \log \spar[f]$.
\end{proof}

\subsection{The log-rank conjecture}

A weak version of the log-rank conjecture for AND-functions, which 
includes an additional $\log n$ factor, now follows quite readily 
from the tools we have developed. 

\logrank*

\begin{proof}  
    By \Cref{lemma:MBS_sparsity}, 
    $\MBS(f) = O((\log r)^2)$. 
    By \Cref{lemma:FMBS_MBS}, $\FMBS(f) = O((\log r)^4)$.
    By \Cref{lemma:0DT_spars}, $\PzeroDT[f] = O((\log r)^5)$.
    By \Cref{claim:0DT_ADT} this gives us 
    an AND-decision tree of height $O((\log r)^5 \cdot  \log n)$. 
    Finally, we convert the AND-decision tree for $f$ into a protocol for
    $\andFunction{f}$ using \Cref{claim:cc_adt} with complexity 
    $O((\log r)^5 \cdot  \log n)$.
\end{proof}

In particular, the log-rank conjecture for AND-functions is true for any
$f$ with $\spar[f] \geq n^c$ for any constant $c>0$. In some sense this is an extremely mild condition,
which random $f$ will satisfy with exceedingly high probability. On the other 
hand, the log-rank conjecture is about \emph{structured} functions; rank and
communication complexity are both maximal for random functions, whereas we 
are interested in low-complexity functions/low-rank matrices. It could very
well be the case that the ultra-sparse regime of $\spar(f) = n^{o(1)}$ is precisely
where the log-rank conjecture \emph{fails}. We therefore see removing
the $\log n$ factor as an essential problem left open by this work.
See \Cref{sec:discussion} for additional discussion. 

\subsection{Lifting AND-functions}
\label{subsection:lifting}

Since $\log(\spar(f))$ lower bounds the deterministic communication 
of $\andFunction{f}$, the log-rank result from the previous section 
immediately implies a new upper bound on the AND decision tree complexity
of $f$. We can prove a better upper bound by making use of our stronger 
assumption: instead of only assuming $\log(\spar(f))$ is small, we assume 
that $\PCC[\andFunction{f}]$ is small. 

If $f$ has large monotone block sensitivity, then its AND-function embeds
unique disjointness as a sub-function. The unique disjointness function 
on $k$ bits, denoted $\UDISJ{k}$, takes two inputs $a, b \in \bool^k$, and is defined
as the partial function:
\[
    \UDISJ{k}(a,b) = 
    \begin{cases}
        0 & \text{if } |a \land b| = 1 \\
        1 & \text{if } |a \land b| = 0 \\
        \text{undefined} & \text{otherwise}
    \end{cases}\;,
\]
where $|\cdot|$ is the Hamming weight. 

\begin{claim}\label{claim:MBS_UDISJ}
    Let $f : \bool^n \to \bool$ be a boolean function with $\MBS(f) = k$.
    Then $\andFunction{f}$ contains as a sub-matrix $\UDISJ{k}$. That is,
    there are  maps $\x, \y : \bool^k \to \bool^n$ and $c \in \bool$ 
    such that the following holds.
    For any 
    $a, b \in \bool^k$ which satisfy that $|a \land b| \in \bool$, it 
    holds that
    \[
        \UDISJ{k}(a, b) = \andFunction{f}(\x(a), \y(b)) \oplus c.
    \]
\end{claim}

\begin{proof}
    Let $z,w_1, \dots, w_k \in \bool^n$ be pairwise disjoint such that 
    $f(z) \ne f(z \lor w_i)$ for all $i \in [k]$. We may assume without loss of
    generality that $f(z)=1$, otherwise replace $f$ with its negation, and set 
    $c = 1$.

    Assume that Alice and Bob want to solve unique-disjointness on inputs 
    $a, b \in \bool^k$,
    which we identify with subsets of $[k]$. Define
    \[
        \x(a) = z \lor \biglor_{i \in a} w_i, \qquad
        \y(b) = z \lor \biglor_{j \in b} w_j.
    \]
    Observe that
    \[
        \x(a) \land \y(b) = 
        \begin{cases}
            z & \text{if } a \land b = \emptyset\\
            z \lor w_i & \text{if } a \land b = \{i\}.
        \end{cases}
    \]
    Thus we get that $\UDISJ{k}(a,b) = f(\x(a) \land \y(b))$ for all $a,b$.
\end{proof}

It is well known that $\UDISJ{k}$ is hard with respect to several communication
complexity measures such as deterministic, randomized and nondeterministic. 
\begin{theorem}[\cite{goos2018landscape, razborov1992distributional}]\label{thm:disj_lbs}
    For any communication complexity measure $\Delta \in \set{\PCC, \BPCC, \NPCC}$,  
    \[
        \Delta(\UDISJ{k}) = \Omega(k).
    \]
\end{theorem}

We immediately get the following corollary:
\begin{corollary}\label{corollary:PCC_MBS}
    Let $f : \bool^n \to \bool$ be a boolean function and 
    $\Delta \in \set{ \PCC, \BPCC, \NPCC}$ be a communication complexity 
    measure. Then $\MBS(f) = O(\Delta(\andFunction{f}))$.
\end{corollary}
\begin{proof}
    Assume that $\MBS[f] = k$. \Cref{claim:MBS_UDISJ} shows that
    any protocol for $\andFunction{f}$ also solves $\UDISJ{k}$. Hence by
    \Cref{thm:disj_lbs} we have $k = O(\Delta(\andFunction{f}))$.
\end{proof}

Taking $\Delta = \PCC$, we obtain the main theorem of this section:

\lifting*

\begin{proof}
    \Cref{claim:rank_sparsity} gives that 
    $\log \spar[f] = \log \rank[\andFunction{f}] \leq C$.
    By \Cref{claim:MBS_UDISJ}, $\MBS(f) = O(C)$.
    By \Cref{lemma:FMBS_MBS}, $\FMBS(f) = O(C^2)$. 
    Combining this upper bound on $\FMBS$ with 
    the fact that $\log \spar[f] \leq C$, we see,
    by \Cref{lemma:0DT_spars}, that $\PzeroDT[f] = O(C^3)$. 
    Finally, by \Cref{claim:0DT_ADT}, we get that 
    $\PandDT[f] = O(C^3 \cdot \log n)$. 
\end{proof}

\section{Generalizations to non-boolean functions}
\label{sec:generalization}

In this section, we extend our conclusion to general multilinear polynomials 
and set systems. The main observation is that all measures introduced in 
\Cref{sec:prelim} are defined for general real-valued functions. In addition,
both \Cref{lemma:FMBS_MBS} and \Cref{lemma:HS_sparsity} are established for 
real-valued functions. The following theorem holds true as the joint result
of these two lemmas.
\generalhsc*
\begin{proof}
    By \Cref{lemma:FMBS_MBS}, $\FHSC(f)=O(m^2)$. Then by \Cref{lemma:HS_sparsity}, 
    we obtain the claimed bound.
\end{proof}

\subsection{Finite-range functions}
\Cref{lemma:MBS_sparsity} is not true for general multilinear polynomials.
Nevertheless, if we make the assumption that the multilinear polynomial's range
is finite, denote its size by $s$, then we can bound the monotone block 
sensitivity by a polynomial of log-sparsity and $s$.
\begin{lemma}
    Let $f: \bool^n \to S$ be a multilinear polynomial
    where $\spar[f] = r$ and $|S| = s$. Then $\MBS(f) = O(s^2 \log^2r)$.
\end{lemma}
\begin{proof}
    Suppose $\MBS(f) = \MBS(f, z) = k$ for $z \in \bool^n$, and let $a = f(z) \in S$. 
    Define a polynomial $p : \R \to \bool$ such that $p(a)=1$ and $p(b)=0$ for 
    $b \in S \setminus \set{a}$. There exist such a polynomial of degree 
    $\deg(p) = |S| - 1$. Define a boolean function $g:\bool^n \to \bool$ by 
    $g(z) = p(f(z))$. Note that $\MBS(g, z) = k$ and $\spar[g] \le r^{s - 1}$. 
    Then by \Cref{lemma:MBS_sparsity}, we have 
    $k = O(\log^2(\spar[g])) = O(s^2 \log^2r)$.
\end{proof}
Combining it with \Cref{thm:gen_hsc}, 
one can bound the hitting set complexity of finite-range functions.
\maingeneral*

The following example shows that a polynomial dependency on the range size is necessary
in \Cref{thm:maingeneral}.

\begin{example}
    Let $f(x)=x_1 + \dots + x_s$. Then $\spar[f] = s$, the range of $f$ has size 
    $s + 1$, and  $\HSC[f] = s$.
\end{example}

\subsection{Set systems}

\Cref{thm:gen_hsc} can also be interpreted in the language of set system.
\generalsetsystem*
\begin{proof}
Let $f(x)=\sum_{i=1}^r \prod_{j \in S_i} x_j$. Fix $m \ge 1$, and consider first
the case that $\MBS(f)<m$. In this case, by \Cref{thm:gen_hsc},
$\HSC[f] = O(m^2 \log r)$. Note that by construction, if $H$ is a hitting set 
for the monomials of $f$ then $H$ is a hitting set for $\mathcal{F}$. 

The other case is that $\MBS(f) \ge m$. Let $z \in \bool^n$ be such that 
$\MBS[f, z] \ge m$. By definition, this implies that $f_z$ has $m$ minimal 
pairwise disjoint sets, which by \Cref{claim:minimal_monomials} implies that the 
polynomial computing $f_z$ contains $m$ pairwise disjoint monomials. Each such 
monomial is of the form $S_i \setminus T$ for $T = \set{i \ :\  z_i = 1}$.
\end{proof}

\section{Discussion}
\label{sec:discussion}

\subsection{Ruling out the \texorpdfstring{$\log n$}{log n} factor}
Both results about communication complexities of AND-functions 
(\Cref{thm:lifting,thm:logrank}) are not ``tight'' in the sense that
both of them have a $\log n$ factor in the right side of the inequality.
Unfortunately, $n$ can be exponential in sparsity (see \Cref{example:redundant-indexing}).

It is easy to see that if the $\log n$ factor is truly necessary in these
theorems we are very close to refuting the log-rank conjecture. Hence,
we believe that a ``tighter'' version of the log-rank theorem 
(\Cref{thm:logrank}) is true.
\begin{conjecture}
    Let $f : \bool^n \to \bool$ be a boolean function, where $\spar[f] = r$. 
    Then
    \begin{equation*}
        \PandDT[f] \le \poly(\log r).
    \end{equation*}
\end{conjecture}
Note that this conjecture would imply a ``tighter'' version of the 
lifting theorem as well. 


\subsection{Randomized complexity}
The main results of this paper are concerned with the deterministic communication
complexity of AND-functions. However, \Cref{corollary:PCC_MBS} says that the randomized communication
complexity of an AND-function is lower bounded by its monotone block sensitivity. 
The relation between randomized communication complexity and sparsity remains
unclear. We conjecture that the relation between these two measures is the same as 
the proved relation (\Cref{thm:lifting}) between sparsity and \emph{deterministic} 
communication complexity.
\begin{conjecture}
\label{conj:randomized}
    Let $f : \bool^n \to \bool$ be a boolean function. Suppose that 
    $\BPPCC(\andFunction{f}) = C$. Then
    \begin{equation*}
        \log(\spar[f]) \le \poly(C) \cdot \log n.
    \end{equation*}
    In particular, $f$ can be computed by an AND-decision tree of depth
    \begin{equation*}
        \PCC(\andFunction{f}) \le \poly(C) \cdot \log n.
    \end{equation*}
\end{conjecture}

Observe that \Cref{conj:randomized} implies that randomness does not significantly help 
to compute AND-functions.  Concretely, it implies that
\[
    \PCC(\andFunction{f}) \le \poly \left( \BPPCC(\andFunction{f}) \right) \cdot \log n.
\]

Interestingly, the $\log n$ factor in this conjecture is necessary as shown 
by the following example.
\begin{example}[Threshold Functions]
    Let $f : \bool^n \to \bool$ be the threshold function such that
    \[
        f(x) = 1 \iff |x| \ge n - 1.
    \]
    It is clear that $\spar[f] = n + 1$; however, $\BPPCC(f) = O(1)$. Indeed, let
    us consider the following randomized AND-decision tree for $f$: it samples a 
    subset $S \subseteq [n]$ uniformly at random, then output the value of
    \[
        q_S(x) = \left( \bigland_{i \in S} x_i \right) \lor 
        \left( \bigland_{i \notin S} x_i \right).
    \]
    Note that if $|x| \ge n-1$ then $q_S(x)=1$ with probability $1$. If $|x| \le n - 2$,
    let $i,j$ be such that $x_i=x_j=0$. With probability $1/2$ we have 
    $i \in S$, $j \notin S$ or $i \notin S$, $j \in S$, in both cases $q_S(x) = 0$. 
    In order to reduce the error, repeat this for a few random sets $S$. 
\end{example}

\subsection{Sparsity vs coefficients size}

Let $f:\bool^n \to \bool$ and consider the multi-linear polynomial computing $f$, namely 
$f(x) = \sum_s f_s \prod_{i \in s} x_i$. It is well known that the 
coefficients $f_s$ take integer values. In particular, if we denote by 
$\|f\|_1 = \sum |f_s|$ the $L_1$ norm of the coefficients, then we get the
obvious inequality
\[
    \spar[f] \le \|f\|_1.
\]
We note the following corollary of \Cref{thm:logrank}, which shows that $\|f\|_1$
cannot be much larger than $\spar[f]$.

\begin{claim}
    Let $f : \bool^n \to \bool$ and assume that $\spar[f] = r$. 
    Then $\|f\|_1 \le n^{O(\log r)^5}$.
\end{claim}

\begin{proof}
    By \Cref{thm:logrank} we have $\PandDT[f] = d$ for $d = O((\log r)^5 \log n)$. 
    By a similar proof to \Cref{claim:adt_exp_sparsity}, any function $f$ computed
    by an AND-decision tree of depth $d$ has $\|f\|_1 \le 3^d$. The claim follows.
\end{proof}

We conjecture that the gap between sparsity and $L_1$ is at most polynomial.

\begin{conjecture}
    For any boolean function $f$, $\|f\|_1 \le \poly(\spar[f])$.
\end{conjecture}

\bibliographystyle{plain}
\bibliography{references}

\end{document}